\documentclass[11pt]{article}

\usepackage{dblfloatfix}
\usepackage{listings}
\usepackage{algorithm}
\usepackage{algpseudocode}
\usepackage[english]{babel}

\usepackage[letterpaper,top=2.2cm,bottom=2.2cm,left=3cm,right=3cm,marginparwidth=1.75cm]{geometry}

\usepackage[utf8]{inputenc}
\usepackage{mathrsfs}
\usepackage{dsfont}
\usepackage{amsfonts}
\usepackage{complexity}
\usepackage{comment}
\usepackage{latexsym}
\usepackage{amssymb}    
\usepackage{amsmath}
\usepackage{graphicx}
\usepackage[colorlinks=true, allcolors=blue]{hyperref}
\newcommand{\opn}{\operatorname}

	\newcommand{\bbN}{\mathbb{N}}

	\newcommand{\bbZ}{\mathbb{Z}}

	\newcommand{\mcH}{\mathcal{H}}

\newcommand{\mcO}{\mathcal{O}}	\newcommand{\mcP}{\mathcal{P}}

\newcommand{\mcU}{\mathcal{U}}

\newcommand{\eps}{\epsilon}

\newcommand{\abs}[1]{\left|#1\right|}                  

\newcommand{\ceil}[1]{\lceil #1 \rceil}
\newcommand{\floor}[1]{\lfloor #1 \rfloor}

\newcommand{\N}{\mathbb{N}}

\newcommand{\Z}{\mathbb{Z}}

\newcommand{\opsigma}{\operatorname{\sigma}}

\newcommand{\pf}{\operatorname{Pf}}

\renewcommand{\P}{\mathsf{P}}

\renewcommand{\epsilon}{\varepsilon}

\renewcommand{\epsilon}{\varepsilon}
\newcommand{\ignore}[1]{}

\usepackage{babel}
\usepackage{amsthm}
\usepackage{thmtools,thm-restate}
\usepackage{fancyhdr}
\usepackage{lineno}

\newtheorem{definition}{Definition}[section]
\newtheorem{theorem}{Theorem}[section]
\newtheorem{claim}[theorem]{Claim}
\newtheorem{corollary}{Corollary}[theorem]
\newtheorem{lemma}[theorem]{Lemma}

\newtheorem{observation}[theorem]{Observation}

\newtheorem{remark}[theorem]{Remark}

\hypersetup{
	colorlinks=true,
	linkcolor=blue,
	citecolor=blue,
	urlcolor=blue,
	pdftitle={Report},
	pdfauthor={}
}
\usepackage{cleveref}

\begin{document}

\title{Catalytic space-efficient algorithms for Distance and Determinant}
\author{
	Srijan Chakraborty \\
	University of Waterloo \\
	\texttt{srijan.chakraborty@uwaterloo.ca}
	\and
	Samir Datta \\
	Chennai Mathematical Institute \\
	\texttt{sdatta@cmi.ac.in}
	\and
	Kushagra Shandilya \\
	Chennai Mathematical Institute \\
	\texttt{kushagras@cmi.ac.in}
}

\date{}
\maketitle

\begin{abstract}
    Catalytic computation reuses potentially large but \emph{in-use} space
    (or catalytic space)
    which must be restored to its initial contents while the \emph{clean} space
    is small. This allows all problems in $\NL$ and even larger classes
    to be in catalytic logspace [Buhrman et al, STOC2014]. In [Chmel et al, CCC2026]
    the catalytic space for reachability was restricted to sublinear
    space (specifically, $\frac{n}{2^{\Theta({\sqrt{\log{n}}})}}$ space)
    while remaining in polynomial time. This matches the non-catalytic bounds
    from [Barnes et al, SICOMP1998].

    We extend the reachability bound to shortest distance, proving a simultaneous
    sublinear catalytic space and polynomial time bound (henceforth
    called, sublinear $\CLP$). This is the first simultaneous sublinear space and
    polynomial time bound for distance even without the catalytic discipline.
    Moreover, our bound works even for polynomially weighted graphs.

    Along the way we prove simultaneous near linear
    catalytic space ($\mcO(n\log^{O(1)}{n})$) and polynomial time bounds for the
    determinant. We also show some applications to planar matching problems.

    Our starting point was catalytic computation of reachability in directed graphs
    by [Cook-Pyne, ITCS2026]. [Edenhofer, arXiv:2602.21088] used it to give a tradeoff
    between catalytic and non-catalytic space, while [Chmel et al, CCC2026]
    compress the graph to a smaller graph preserving $s,t$-reachability to
    achieve a smaller space bound. All three can be modified to yield an algorithm
    for unweighted distance.

    Our contribution is to define an approximation to the walk polynomial in
    weighted graphs that allows us to infer \emph{weighted} $s,t$-distance in
    sublinear $\CLP$. We do so by non-trivially combining elements from each of
    the above three papers.
\end{abstract}

\newpage

\section{Introduction}
Achieving a logspace bound for various problems is the gold standard in
space bounded computation -- achieved in particular for (undirected) connectivity
by Reingold's breakthrough result \cite{Reingold08}. From a reduction to connectivity, this carried
over to a plethora of problems like planarity testing \cite{AM04,DP},
bounded treewidth decomposition and MSO evaluation therein \cite{EJT},
planar, bounded genus and bounded treewidth isomorphism \cite{DLNTW,EK,ES}.
However for problems that were not provably or even believably
in Logspace such as reachability, the problem of giving a good deterministic
space bounded characterization persisted.

Of course, Savitch's theorem could squeeze
reachability into polylogarithmic space but at the expense of quasipolynomial
running time. Rather than compromise on polynomial time, there were attempts
to limit the non-determinism by considering
unambiguous computation and Reinhardt-Allender \cite{RA} were able to prove
that Reachability is in fact in unambiguous logspace if only at the expense of
non-uniformity. Bourke et al \cite{BTV}
were even able to get rid of
the non-uniformity for planar reachability. However non-deterministc computation
even if unambiguous is not quite deterministic.

This clarified what was
required: a polynomial time, deterministic algorithm that worked in low space.
Barnes et al \cite{BBRS} achieved a deterministic
sublinear space algorithm for reachability, that was simultaneously polynomial time bounded.
Here the matters rested for a long time.

Then in a master stroke, building upon an old idea of Ben-Or \& Cleve \cite{BC}, Buhrman et al \cite{BCKLS} introduced the class $\CL$. The idea was to catalyse Logspace with
the use of a large \emph{filled} space that had to be restored to its original
contents before the algorithm terminated. Suddenly we were back in the domain
of deterministic space computation, more so when Kouck\'y et al. \cite{KMPS}
proved that adding non-determinism or randomization to this model does not
make it more powerful. In the original paper itself it was shown that not only
reachability but a significant part of linear algebra and in fact, all of $\TC^1$ are in
$\CL$.

The one flip side of the great power of the class $\CL$ is that it isn't
even known to be polynomially bounded on all inputs i.e. inside $\P$. However
Buhrman et al \cite{BCKLS} were able to prove a $\ZPP$ bound by viewing $\CL$ as a
probabilistic class which is polynomial time bounded in the expectation.

Soon catalytic logspace was populated with problems that had long defied being
classified
as deterministically parallelisable -- the prominent examples being
bipartite matching , linear matroid intersection\footnote{The recent breakthrough of Chatterjee et al \cite{CGGRT26} puts Bipartite Matching and Linear matroid intersection in $\NC$.} and even
general matching \cite{AM,AAV26,CDKMS}. Surprisingly these bound were simultaneously
polynomially bounded i.e. in $\CLP$.

Having achieved $\CLP$, a sceptic would tend to believe that this was where the
story ended.
However, soon enough, contrary to this scepticism the efforts to save
catalytic space
for $\CLP$ problems emerged.
A prominent example is the
$O(\log^2{n})$ bound on catalytic space simultaneous with logarithmic workspace
and polynomial time, for $\BPL$ given by Pyne in \cite{Pyne}. Notice that this
catalytically achieved the upper bound of $\SC^2$ proved by Nisan on $\BPL$
\cite{Nis94}.

A similar bound on reachability would be the natural next step. Preliminary
efforts mimicking Pyne \cite{Pyne} were made in Arvind et al \cite{ACD} using unambiguity to
give a tradeoff between space and catalytic space, which was catalytically
retracing the footsteps of Reinhardt-Allender \cite{RA}.
In the deterministic setting, this effort
is stymied at the outset by the difficulty in improving the space bound of
reachability beyond the Barnes et al \cite{BBRS} bound of slightly
sublinear, $\frac{n}{2^{\Theta(\sqrt{\log{n}})}}$, along with simultaneous
polynomial time even in the non-catalytic setting. Yet, after initial
catalytic space versus time tradeoff results by Cook \& Pyne \cite{CP} and
space versus catalytic space tradeoff results by Edenhofer \cite{Edenhofer} the
effort culminates in the frontier paper of Chmel et al \cite{CDKMR} that achieves the
Barnes et al \cite{BBRS} space versus time tradeoff catalytically i.e. works in sublinear
catalytic space while simultaneously using logarithmic workspace and
polynomial time.

This is where our results enter the fray. We show that not only reachability
but also iterated matrix product of $n$ matrices each of dimension $n$
(this was implicit in \cite{BCKLS}), directed distance in $n$ vertex graphs
and even the determinant of an $n \times n$ matrix can be computed in
near linear $\CLP$ i.e. catalytic logspace using $O(n\log^{O(1)}{n})$ space and
simultaneous polynomial time -- in fact we show a trade off result between workspace and catalytic space, similar to \cite{Edenhofer}.

We then go on to show that distance in $n$ vertex, polynomially bounded edge weighted graphs can be computed in
sublinear $\CLP$ i.e.  $n/2^{\sqrt{\log{n}}}$ catalytic space with simultaneous
logarithmic workspace and polynomial time. Even without catalytic
space, as far as we are aware this is the first sublinear space bound for
weighted distance while staying in polynomial time. Any improvements to this
bound on distance would naturally yield corresponding improvements for
reachability.

For the determinant result our techniques build on Cook-Pyne \cite{CP} along with
an identity that relates the generating function of elementary symmetric
polynomials to the power polynomials
along with polynomial interpolation from Hesse et al \cite{HAB}. For the distance result
we build on the frontiers paper Chmel et al \cite{CDKMR} and combine it with an IMM$_n$
tradeoff \emph{a la} Edenhofer \cite{Edenhofer}.

We also exhibit two matching applications of our results.
The first is search and decision for planar perfect matchings in near linear $\CLP$,
as a consequence of being able to compute the
Pfaffian (analogous to determinant computation) in near linear $\CLP$.
The applications combine our
main results with techniques from various sources \cite{Kastelyn,MSV,ACD}.
The second is for the restriction to planar bipartite graphs, we can further improve the decision and search for perfect matching to sub linear $\CLP$ which uses a technique from Miller-Naor\cite{MN}.

\subsection{Proof Idea}
We discuss the essential ideas behind our results.
\subsubsection{Tradeoff between space and catalytic space}
Cook \& Pyne \cite{CP} give simple catalytic algorithms for directed
reachability that
not only are time efficient but are catalytic space efficient as well, using
roughly quadratic space\footnote{They also describe linear space randomized
	algorithms.}. Their technique (called somewhat picturesquely ``stirring
the pot'' by Nathan Sheffield in an expository essay \cite{Sheff}) consists of
associating a catalytic register with each vertex -- these together constitute
the ``pot''. Now ``stirring'' this consists of translating the value of each
register by the current value of each of its in-neighbours in some arbitrary
order several times. Unstirring is the process of stirring the pot in an
order reverse of the initial one and restores the original values if done
exactly the same number of times as stirring. Their crucial observation is
that if the value of the source register is changed initially then after
stirring, the value obtained at the destination vertex
changes iff there is a path from the source vertex $s$ to the
destination vertex $t$.
Thus, stirring followed by saving the value at $t$ followed by
unstirring followed by changing the value at $s$ by $1$, stirring again
and comparing the stored value at $t$ with the new value
before unstirring and restoring the value at $s$ catalytically reveals the
presence of $s,t$-reachability.

Our first observation is that the stirring technique, applied to a layered version of the original graph, computes a polynomial (if we offset the value at a vertex by $x$ times each predecessor's
value) that we call the \emph{stirring polynomial} and further this is identical
with the generating function of walks aka \emph{walk polynomial} in the original graph.

The stirring technique in \cite{CP} has been developed by Edenhofer
\cite{Edenhofer}
to give a work space vs. catalytic space tradeoff for reachability as well.
We mimic their construction to give a similar tradeoff for matrix powering
(which is basically a tabulation of walk polynomials)
and generalise it to iterated matrix product or equivalently  branching
program evaluation. This tradeoff yields a near linear catalytic space
simulation of the matrix power of $n$ while restricting to logarithmic work
space and polynomial time\footnote{Henceforth abbreviated to near linear
	$\CLP$.}. Our next observation is that this yields determinant
of an $n \times n$ matrix in near linear $\CLP$ using an identity
relating the generating function of elementary symmetric polynomials the
power polynomials.

On the other hand we show that it also yields a polynomial time and
logarithmic space algorithm for powering matrices to small powers,
$2^{\Theta(\sqrt{\log{n}})}$ while staying in sublinear catalytic space,
$\mcO(\frac{n}{2^{\Theta(\sqrt{\log{n}})}})$ \footnote{Henceforth
	any unqualified use of the word sublinear will refer to
	$\mcO(\frac{n}{2^{\Theta(\sqrt{\log{n}})}})$ and short/small will refer to
	$2^{\Theta(\sqrt{\log{n}})}$.}. Notice that short distances
can be computed in these bounds as well since the smallest
non-zero exponent in a walk polynomial of a small power will yield that.

\subsubsection{Sublinear space algorithms}
Next we turn to the recent publication \cite{CDKMR} that leverages the
techniques in \cite{CP} to give a sublinear  $\CLP$ bound for
reachability mirroring similar uncatalysed bound from \cite{BBRS} which
has not been improved for nearly three decades. \cite{CDKMR} construct a
sparsified version of the input graph where linear length walks are replaced
by sublinear walks preserving reachability where the edges
are related to walks of length $2^{\Theta(\sqrt{\log{n}})}$.

We are able to modify the construction of \cite{CDKMR} so that
we store distance information on the short-walk-edges. Since we already had
(weighted) distance in near linear space this works out to sublinear space
when run on the sparsified graph of \cite{CDKMR}. This includes being able
to compute the weights corresponding to the distances in short walks in
sublinear space.

This yields the first sublinear $\CLP$ algorithm for distance
in a weighted (polynomially bounded) directed graph. In fact, this is true even if we ignore the word
catalytic; it remains the first sublinear space algorithm for distance in weighted
digraphs which is simultaneously in polynomial time.

\subsection{Applications}
We now turn to some applications of our results concerning matrix powering and distance. We first establish determinant computation in near linear $\CLP$ using an identity which relates generating function of elementary symmetric polynomials and power polynomials. In a way similar to the determinant we show Pfaffian computation of a matrix in near linear $\CLP$.

We also show two graph matching applications. The first is the search and decision version of planar perfect matching in near linear $\CLP$. The weighted decision version uses a space efficient version of Kastelyn's algorithm and Pfaffian computation based on \cite{MSV}. The search version is handled using a careful application of catalytic isolation \cite{ACD}.

We also show planar bipartite perfect matching in sublinear $\CLP$. For this we use the Miller and Naor algorithm \cite{MN}, that reduces computation of flows to distances in planar graphs.

\subsection{Organization}
We start with preliminaries in \Cref{sec:prelims}. We establish space-space tradeoff results for matrix powering and branching program in \Cref{sec:mpbp}. Using these results we prove weighted distance in sublinear $\CLP$ in \Cref{sec:subd}. In \Cref{sec:app}, we give applications concerning determinant and matchings in graphs. Some of the proofs in \Cref{sec:mpbp,sec:subd,sec:app} have been moved to the Appendix.

\section{Preliminaries} \label{sec:prelims}

\begin{definition}[Catalytic Computation]
        A catalytic Turing machine (CTM) $M$ works as follows. The machine $M$ consists of three tapes: a read-only input tape, a read-write work tape and a read-write catalytic tape containing some arbitrary string $\tau$. For some language $L$ and input $x$, $M$ halts, outputs $x \in L$ and the catalytic tape contains $\tau$.

        A language $L \in \CTISP[t, w, c]$ if there exists a CTM which decides membership in $L$ and takes $t$ time, $w$ bits for work tape and $c$ bits for catalytic tape.

        A language $L$ is in Catalytic Logspace, denoted by $\CL$, if the CTM for $L$, for input of size $n$, takes $\mcO(\log n)$ work space and $\mcO(n^{\mcO(1)})$ catalytic space.

        We denote $\CLP$, sublinear $\CLP$ and near-linear $\CLP$ as the set of languages which are in $\CL$, their corresponding machines runs in polynomial time and takes catalytic space $n^{\mcO(1)}$, $\dfrac{n}{2^{\Theta(\sqrt {\log n})}}$ and $n \log^{\mcO(1)}n$ respectively.
\end{definition}

\begin{definition}[Pfaffian]
        Let $A$ be a skew symmetric matrix of even size $n \times n$ which is equivalent to an orientation of an undirected graph $G$ over $n$ vertices. We define \emph{perfect matching} $M$ to be a collection of edges as $M = \{(\opsigma(1), \opsigma(2)), (\opsigma(3), \opsigma(4)) \dots, (\opsigma(n-1), \opsigma(n))\}$ where $\opsigma \in S_n$.

        Since $A$ is a skew symmetric matrix, many permutations corresponds to the same perfect matching. Thus, for any perfect matching $M$, the canonical permutation, $\opsigma_M$, is the permutation in which edges are from smaller to larger vertices and edges are listed in increasing order of smaller vertices in each edge, i.e., $\opsigma_M(2i-1) < \opsigma_M(2i)$ for $i = 1, 2, \dots, n/2$, and $\opsigma_M(1) < \opsigma_M(3) < \dots < \opsigma_M(n-1)$. We define $\operatorname{wt}(\opsigma_M) = \prod_{i=1}^{n/2}A[\opsigma_M(2i-1), \opsigma_M(2i)]$.

        The Pfaffian of $A$ (or equivalently $G$) is defined as
        \[
                \pf(A) = \sum_{M \in \mathcal{M}} \operatorname{sgn}(\opsigma_M) \operatorname{wt}(\opsigma_M),
        \]
        where $\mathcal{M}$ is the set of all perfect matchings.

\end{definition}
Cayley's theorem states that for a skew symmetric matrix $A$ of even size, $\det(A) = (\pf(A))^2$. This gives us $\abs{\pf(A)} = (\det(A))^{\frac{1}{2}}$.

\section{Matrix Powering and Branching Programs} \label{sec:mpbp}
In this section we give space-space tradeoff algorithms for matrix powering and branching programs.

\subsection{Matrix Powering}
We first present the space-space tradeoff result for matrix powering.

\begin{theorem}\label{thm:apow}
	Given an $n\times n$ matrix $A$, and indices $s, t \in \Z_n$, a parameter $\ell \in \N^+$ and a parameter $k \in[1, n]$, there is an algorithm running in time $\poly(n,(3k)^{\log \ell})$.
	, work space $\mcO(\log \ell\cdot \log{(3k)}+\log n)$ and catalytic space $\mcO(\frac n k \cdot \log n \cdot \log \ell)$, that outputs $A^\ell[s,t]$. In particular, For any $k \in[1, n]$, $A^\ell[s,t]$ can be computed in
	
	$$
		\C\TISP\left[\poly(n,(3k)^{\log \ell}),\mcO(\log \ell \cdot \log{(3k)}+\log n), \mcO\left(\frac{n}{k} \cdot \log n \cdot \log \ell \right)\right].
	$$
\end{theorem}

Let $k \in [1, n]$ be a parameter. We partition the set of vertices, $\Z_n$, into $k$ parts each of size $n/k$ (assume, for brevity that $k$ divides $n$) as $\Z_n=\sqcup_{i\in\bbZ_k}\{v\in \Z_n:v\equiv i\mod k\}$. Each index $v$ can be uniquely written as $\lambda_vk+r_v$ where $\lambda_v=\floor{v/k},r_v\in \bbZ_k$.

Fix a prime modulus $q$ with $\log q = \mcO(\log n)$. The algorithm will use
$\lceil \log \ell \rceil + 2$ blocks of registers each on the catalytic tape, denoted
$U, V, W_1, \ldots, W_{r = \lceil \log \ell \rceil}$, and each block will contain $n/k$ registers. Each register stores an element of $\mathbb{Z}_q$ as in the encoding defined in \cite[Appendix~A]{Edenhofer} and so the total catalytic tape is $\mcO((n/k) \log n \log \ell)$ bits. 

Consider the register programs $P^\ell_{i,j}$ for $\ell \in \bbN^+, i,j \in \{0, \dots, k-1\}$, as described in \Cref{nl:P23}.
\begin{algorithm}
	\caption{$P^{\ell}_{i,j}(U,V,W_1,\ldots,W_{r=\ceil{\log \ell}})$ for $i,j \in \{0, \dots, k-1\}$}\label{nl:P23}
	\begin{algorithmic}[1]
		\State \textbf{Input:} $A$.
		\If {$\ell=1$}
		\For {$u,v\in\Z_n$ such that $r_u=i,r_v=j$}
		\State $V[\lambda_v]=V[\lambda_v]+A[u,v]\cdot U[\lambda_u] \pmod{q}$
		\EndFor
		\Else
		\For {$m \in \{0, \dots, k-1\}$}
		\State $P^{\ceil{\ell/2}}_{i,m}(U,W_r,W_1,\ldots,W_{r-1})$ \Comment{Call this step $1$}
		\State $P^{\floor{\ell/2}}_{m,j}(W_r,V,W_1,\ldots,W_{r-1})$ \Comment{Step $2$}
		\State $P^{\ceil{\ell/2}}_{i,m}(U,W_r,W_1,\ldots,W_{r-1})^{-1}$\Comment{Step $3$: The inverse of step $1$ i.e line $8$.}
		\EndFor
		\EndIf
	\end{algorithmic}
\end{algorithm}

The programs $P^{\ell}_{i,j}$ are reversible. $(P^{\ell}_{i,j})^{-1}$ can be obtained by two changes in the above algorithm, replacing "+" with "-" in line 4 and $P^{\floor{\ell/2}}_{m,j}$ with $(P^{\floor{\ell/2}}_{m,j})^{-1}$ in line 9.

Fix $\ell \in \mathbb{N}^+$ and $i, j \in \{0, \dots, k-1\}$. For a vector $\tau=\left(\tau_0, \ldots, \tau_{(n / k)-1}\right) \in \mathbb{Z}_q^{n / k}$ and an index $v \in \Z_n$ with $v \equiv j(\bmod k)$, let $\alpha_{\ell, v}^{(i, j)}(\tau) \in \mathbb{Z}_q$ denote the value added to the register $V\left[\lambda_v\right]$ by executing $P_{i, j}^{(\ell)}\left(U, V, W_1, \ldots, W_r\right)$ when the input block $U$ is initialized so that for every index $u \equiv i(\bmod k), U\left[\lambda_u\right]=\tau_{\lambda_u}$.

\begin{claim} \label{claim:matpowregunchanged}
	For every $\ell \in \bbN^+$ and $i, j \in \{0, \dots, k-1\}$, executing $P^{\ell}_{i,j}(U,V,W_1,\ldots,W_r)$ leaves all registers unchanged except those in block $V$.
\end{claim}
The proof of the above follows the same as of \cite{Edenhofer} Claim 3, thus we refer to that for more details.

\begin{claim}[Matrix powering analogue of \cite{Edenhofer} Claim 9]\label{claim:cl11}
	
	Let $A \in \mathbb{Z}_q^{n \times n}$. For every $\ell \in \mathbb{N}^+$, $i,j \in \{0, \dots, k-1\}$, vectors $\tau, b \in \mathbb{Z}_q^{n/k}$, and every vertex $v \in \Z_n$ with $v \equiv j \pmod{k}$,
	\[
		\alpha_{\ell, v}^{(i, j)}(\tau+b)-\alpha_{\ell, v}^{(i, j)}(\tau)
		\equiv
		\sum_{\substack{u \in \Z_n :\\ u \equiv i \,(\mathrm{mod}\, k)}}
		b_{\lambda_u}\,(A^\ell)_{u,v}
		\pmod{q}.
	\]
	
	In particular, if $b_{\lambda_s}=1$ and $b_{\lambda}=0$ for all $\lambda\ne \lambda_s$, then
	\[
		\alpha_{\ell, v}^{(i=r_s, j)}(\tau+b)-\alpha_{\ell, v}^{(i=r_s, j)}(\tau)
		\equiv (A^\ell)_{s,v} \pmod{q}.
	\]
\end{claim}

\begin{proof}
	We proceed by induction on $\ell$.
	
	\textbf{Base case $\ell=1$.}
	By definition of $P^{1}_{i,j}$, for every $v$ with $r_v=j$,
	\[
		\alpha^{(i,j)}_{1,v}(\tau)
		=
		\sum_{\substack{u \in \Z_n :\\ r_u = i}}
		A[u,v]\cdot \tau_{\lambda_u}.
	\]
	Hence,
	\[
		\alpha^{(i,j)}_{1,v}(\tau+b) - \alpha^{(i,j)}_{1,v}(\tau)
		=
		\sum_{\substack{u \in \Z_n \\ r_u = i}}
		b_{\lambda_u} A[u,v].
	\]

	\textbf{Inductive step.}
	Assume the statement holds for all lengths strictly smaller than $\ell \ge 2$, and let
	$r=\lceil \log \ell \rceil$.
	
	Fix $m \in \{0, \dots, k-1\}$. By definition of $P^{\ell}_{i,j}$, the first recursive call
	$
		P^{\lceil \ell/2 \rceil}_{i,m}(U,W_r, 	W_1,\ldots,W_{r-1})
	$
	produces intermediate contributions stored in $W_r$ corresponding to length-$\lceil \ell/2 \rceil$
	matrix propagation from residue $i$ to residue $m$. The second call
	$
		P^{\lfloor \ell/2 \rfloor}_{m,j}(W_r,V, W_1,\ldots,W_{r-1})
	$
	uses $W_r$ to add contributions corresponding to length-$\lfloor \ell/2 \rfloor$ propagation
	from residue $m$ to residue $j$. The third call uncomputes $W_r$ and does not affect $V$.
	
	Let $\sigma \in \mathbb{Z}_q^{n/k}$ denote the initial contents of $W_r$. Then
	\[
		\alpha^{(i,j)}_{\ell,v}(\tau)
		\equiv
		\sum_{m \in \Z_k}
		\alpha^{(m,j)}_{\lfloor \ell/2 \rfloor, v}
		\big(\sigma + \alpha^{(i,m)}_{\lceil \ell/2 \rceil,\cdot}(\tau)\big)
		\pmod{q}.
	\]
	where $\alpha^{(i,m)}_{\lceil \ell/2 \rceil,\cdot}(\tau)$ -- notice the dot in the subscript of $\alpha^{(i,m)}_{\lceil \ell/2 \rceil,\cdot}(\tau)$ -- denotes a vector in
	$\mathbb{Z}_q^{n/k}$ indexed by the block variable $\lambda_u \in [n/k]$. Therefore,
	\begin{align*}
		\alpha^{(i,j)}_{\ell,v}(\tau+b) - \alpha^{(i,j)}_{\ell,v}(\tau)
		 & \equiv
		\sum_{m \in \Z_k}
		\Big(
		\alpha^{(m,j)}_{\lfloor \ell/2 \rfloor, v}
		(\sigma + \alpha^{(i,m)}_{\lceil \ell/2 \rceil,\cdot}(\tau+b)) \\
		 & \quad -
		\alpha^{(m,j)}_{\lfloor \ell/2 \rfloor, v}
		(\sigma + \alpha^{(i,m)}_{\lceil \ell/2 \rceil,\cdot}(\tau))
		\Big)
		\pmod{q}.
	\end{align*}
	
	Applying the induction hypothesis to the outer level gives
	\[
		\equiv
		\sum_{m \in \Z_k}
		\sum_{\substack{u \in \Z_n :\\ r_u = m}}
		\Big(
		\alpha^{(i,m)}_{\lceil \ell/2 \rceil, u}(\tau+b)
		-
		\alpha^{(i,m)}_{\lceil \ell/2 \rceil, u}(\tau)
		\Big)
		(A^{\lfloor \ell/2 \rfloor})_{u,v}
		\pmod{q}.
	\]
	
	Applying the induction hypothesis again to the inner term yields
	\[
		\equiv
		\sum_{m \in \Z_k}
		\sum_{\substack{u \in \Z_n :\\ r_u = m}}
		\sum_{\substack{w \in \Z_n :\\ r_w = i}}
		b_{\lambda_w}
		(A^{\lceil \ell/2 \rceil})_{w,u}
		(A^{\lfloor \ell/2 \rfloor})_{u,v}
		\pmod{q}.
	\]
	
	Reordering sums, we get
	\[
		\equiv
		\sum_{\substack{w \in \Z_n :\\ r_w = i}}
		b_{\lambda_w}
		\sum_{u \in \Z_n}
		(A^{\lceil \ell/2 \rceil})_{w,u}
		(A^{\lfloor \ell/2 \rfloor})_{u,v}
		\pmod{q}.
	\]
	
	Using matrix multiplication,
	\[
		\sum_{u \in \Z_n}
		(A^{\lceil \ell/2 \rceil})_{w,u}
		(A^{\lfloor \ell/2 \rfloor})_{u,v}
		=
		(A^\ell)_{w,v}.
	\]
	
	Thus,
	\[
		\alpha^{(i,j)}_{\ell,v}(\tau+b) - \alpha^{(i,j)}_{\ell,v}(\tau)
		\equiv
		\sum_{\substack{w \in \Z_n :\\ r_w = i}}
		b_{\lambda_w}(A^\ell)_{w,v}
		\pmod{q},
	\]
	completing the proof.
\end{proof}

We now give the algorithm for computing $A^\ell[s, t]$ and subsequently prove \Cref{thm:apow}.

\begin{algorithm}
	\caption{$A^\ell[s,t]$ modulo $q$}\label{nl:alg23}
	\begin{algorithmic}[1]
		\For{$c \in \{0,1\}$}
		\State $U[\lambda_s] \leftarrow U[\lambda_s] + c$
		\State run $P^{\ell}_{r_s,r_t}(U,V,W_1,\ldots,W_r)$
		\State $\alpha_c \leftarrow V[\lambda_t]$
		\State run $P^{\ell}_{r_s,r_t}(U,V,W_1,\ldots,W_r)^{-1}$
		\State $U[\lambda_s] \leftarrow U[\lambda_s] - c$
		\EndFor
		\State return $\alpha_1 - \alpha_0 \pmod{q}$
	\end{algorithmic}
\end{algorithm}

\begin{proof}[Proof of \Cref{thm:apow}]
	For a fixed parameter $k$, and prime $q$, notice that \Cref{nl:alg23} computes $A^{\ell}[s,t]\pmod{q}$ using catalytic space $(n/k)\log \ell \log q=\mcO((n/k)\log{\ell} \log n)$. Further, observe that the recursion in \Cref{nl:P23} has branching $k$, and recursion depth $\log \ell$, thus the time taken is $\poly(n,(3k)^{\log \ell})$, and workspace used is $\mcO(\log n+ \log \ell \log{(3k)})$. Computing $A^\ell[s,t]\pmod{q}$ for various primes $q$, and applying chinese remaindering in logspace through \cite{HAB}, gives us the desired result.
\end{proof}

\subsection{Branching Program}
The space-space tradeoff result for branching program is as follows.

\begin{theorem}[Catalytic evaluation of branching programs]\label{thm:abp}
	Let $G$ be a layered branching program of length $\ell$ over vertex set $\Z_n$, where each layer $\ell_i \in \{1, \dots ,\ell\}$ is described by a transition matrix
	$
		A^{(\ell_i)} \in \mathbb{Z}_q^{n\times n}.
	$
	Given start state $s$ and target state $t$, and a parameter $k \in [1, n]$, there is an algorithm computing the weight
	\[
		\sum_{\pi: s \leadsto t} \prod_{r=1}^{\ell} A^{(r)}_{\pi_{r-1},\pi_r}
		\;=\;
		(A^{(1)}A^{(2)}\cdots A^{(\ell)})[s,t]
	\]
	in time $\poly(n,(3k)^{\log \ell})$, work space $\mathcal{O}(\log n + \log (3k) \cdot \log \ell)$, and catalytic space $\mathcal{O}\!\left(\frac{n}{k}\log n \log \ell \right)$.

	\begin{algorithm}
		\caption{$P^{\ell}_{i,j}(U,V,W_1,\ldots,W_r)[A^{(1..\ell)}]$ for layered branching programs}
		\label{alg:BP-P}
		\begin{algorithmic}[1]
			\State \textbf{Input:} layer sequence $A^{(1)},\dots,A^{(\ell)}$
			\If {$\ell=1$}
			\For {$u,v \in \Z_n$ with $r_u=i,\ r_v=j$}
			\State $V[\lambda_v] \leftarrow V[\lambda_v] + A^{(1)}[u,v]\cdot U[\lambda_u] \pmod{q}$
			\EndFor
			\Else
			\For {$m \in \Z_k$}
			\State run $P^{\lceil \ell/2\rceil}_{i,m}( U, W_r,W_1,\ldots,W_{r-1})[A^{(1..\ell_1)}]$\Comment{$\ell_1=\ceil{\ell/2}$}
			\State run $P^{\lfloor \ell/2\rfloor}_{m,j}( W_r, V, W_1, \ldots,W_{r-1})[A^{(\ell_1+1..\ell)}]$
			\State uncompute $W_r$
			\EndFor
			\EndIf
		\end{algorithmic}
	\end{algorithm}
	
\end{theorem}

Here, for brevity of notation, we denote $A^{(i..j)}=A^{(i)}\ldots A^{(j)}$ for $i<j$. We remind from the previous discussion that we partition $[0, n-1]$ as $[0, n-1]=\sqcup_{i\in\bbZ_k}\{v\in [0, n-1]:v\equiv i\mod k\}$ and each index $v$ can be uniquely written as $\lambda_vk+r_v$ where $\lambda_v=\floor{v/k},r_v\in \bbZ_k$.

We again note that the programs $P^{\ell}_{i,j}$ are reversible and executing $P^{\ell}_{i,j}(U,V,W_1,\ldots,W_r)$ leaves all registers unchanged except those in block $V$, as mentioned in \ref{claim:matpowregunchanged} and their proofs are analogous. 

\begin{claim}[Branching program analogue of \cite{Edenhofer} Claim 9]
	Let $A^{(1)},\dots,A^{(\ell)}$ be a layered branching program.
	For all $\ell \ge 1$, $i,j \in \{0, \dots, k-1\}$, vectors $\tau,b \in \mathbb{Z}_q^{n/k}$, and $v \equiv j \pmod{k}$,
	\[
		\alpha^{(i,j)}_{\ell,v}(\tau+b)-\alpha^{(i,j)}_{\ell,v}(\tau)
		\equiv
		\sum_{\substack{u \in \Z_n\\ r_u=i}}
		b_{\lambda_u}\cdot (A^{(1)}A^{(2)}\cdots A^{(\ell)})_{u,v}
		\pmod{q}.
	\]
\end{claim}

\begin{proof}
	We proceed by induction on $\ell$.
	
	\textbf{Base case $\ell=1$.}
	Immediate from linearity of the transition step:
	\[
		\alpha^{(i,j)}_{1,v}(\tau)
		\equiv
		\sum_{\substack{u \in \Z_n:\\ r_u=i}}
		A^{(1)}[u,v]\cdot \tau_{\lambda_u}
		\pmod{q}.
	\]
	
	\textbf{Inductive step.}
	Let $\ell_1=\lceil \ell/2\rceil$, $\ell_2=\lfloor \ell/2\rfloor$.
	By definition of layered branching programs, any length-$\ell$ path decomposes uniquely at layer $\ell_1$, giving
	$
		A^{(1..\ell)} = A^{(1..\ell_1)} \cdot A^{(\ell_1+1..\ell)}.
	$
	
	Thus the recursive computation satisfies,
	\[
		\alpha^{(i,j)}_{\ell,v}(\tau)
		\equiv
		\sum_{m \in \Z_k}
		\alpha^{(m,j)}_{\ell_2,v}
		\big(\sigma + \alpha^{(i,m)}_{\ell_1,\cdot}(\tau)\big)
		\pmod{q},
	\]
	where $\alpha^{(i,m)}_{\lceil \ell/2 \rceil,\cdot}(\tau)$ denotes a vector in $\mathbb{Z}_q^{n/k}$ indexed by the block variable $\lambda_u \in [n/k]$.Therefore,
	\begin{align*}
		\alpha^{(i,j)}_{\ell,v}(\tau+b) - \alpha^{(i,j)}_{\ell,v}(\tau)
		 & \equiv
		\sum_{m \in \Z_k}
		\Big(
		\alpha^{(m,j)}_{\ell_2,v}
		\big(\sigma + \alpha^{(i,m)}_{\ell_1,\cdot}(\tau + b)\big)
		-
		\alpha^{(m,j)}_{\ell_2,v}
		\big(\sigma + \alpha^{(i,m)}_{\ell_1,\cdot}(\tau)\big)
		\Big)
		\pmod{q}.
	\end{align*}
	
	Applying the induction hypothesis successively twice to the inner term, we get
	\[
		\equiv
		\sum_{m \in \Z_k}
		\sum_{\substack{u\in\Z_n:\\ r_u=m}}
		\sum_{\substack{w\in\Z_n:\\ r_w=i}}
		b_{\lambda_w}
		(A^{(1..\ell_1)})_{w,u}
		(A^{(\ell_1+1..\ell)})_{u,v} \pmod{q}.
	\]
	
	Reordering sums,
	\[
		\equiv
		\sum_{\substack{w\in\Z_n:\\ r_w=i}}
		b_{\lambda_w}
		\sum_{u\in\Z_n}
		(A^{(1..\ell_1)})_{w,u}(A^{(\ell_1+1..\ell)})_{u,v} \pmod{q}.
	\]
	
	By matrix multiplication,
	\[
		\equiv
		\sum_{\substack{w\in\Z_n:\\ r_w=i}}
		b_{\lambda_w}
		(A^{(1..\ell)})_{w,v} \pmod{q}.
	\]

	Hence,
	\[
		\alpha^{(i,j)}_{\ell,v}(\tau+b)-\alpha^{(i,j)}_{\ell,v}(\tau)
		\equiv
		\sum_{\substack{w \in \Z_n:\\ r_w=i}}
		b_{\lambda_w}\cdot (A^{(1)}A^{(2)}\cdots A^{(\ell)})_{w,v} \pmod{q}.
	\]
\end{proof}
With the above proved, the proof of \Cref{thm:abp} follows in the same way as the proof of \Cref{thm:apow}.

\section{Sublinear Catalytic Space for Distances}\label{sec:subd}
In this section we prove that distance in polynomially weighted graphs can be computed in sublinear $\CLP$. The main result is stated in the following.

\begin{theorem}\label{thm:dist}
	Given an $n$ vertex directed graph $G=(V,E)$ with polynomially bounded edge weights $W:E\to \bbZ$, and any two designated vertices
	$s,t\in V$, we can compute the minimum weight of an $s$-$t$ walk in sublinear $\CLP$,
	assuming there is no negative cycle in the graph.
\end{theorem}

From \Cref{thm:apow}, we know that we can compute $A^\ell$ where $A$ is an
$n\times n$ matrix and $\ell\le n$, in
$\C\TISP\Big[\poly(n,k^{\log \ell}),
		\mcO(\log \ell\cdot \log k+\log n),
		\mcO\Big(\dfrac n k \cdot \log^2 n\Big)\Big].$

Therefore, we can compute $A^{2^{\sqrt{\log n}}}$ in
$\C\TISP\Big[\poly(n),\mcO(\log n),
		\mcO\Big(\dfrac{n}{2^{\Theta(\sqrt {\log n})}} \Big)\Big]$
(by setting $k=2^{\mcO(\sqrt{\log n})}$). Thus we have the following lemma for short length walks:

\begin{lemma}\label{short}
	Given a directed graph $G$ on $n$ vertices (without self-loops) 
	with polynomially bounded edge weights, we can compute the minimum weight of
	a walk from $s$ to $t$ that uses at most $2^{c\sqrt{\log n}}$ edges in sublinear $\CLP$
	for any constant $c$.
\end{lemma}

\begin{proof}
	Let $A$ be the weighted adjacency matrix over $\mathbb Z[x]$ given by
	\[
		A[u,v]=
		\begin{cases}
			x^{w(u,v)} & (u,v)\in E,       \\[2mm]
			1          & u=v               \\[2mm]
			0          & \text{otherwise.}
		\end{cases}
	\]
	
	We can compute $A^\ell$ in\footnote{Notice that the exponents in the
		entries of $A$ may be negative but are upper bounded by some $W_{\max}$ in 
		magnitude.  Thus if we multiply each entry of $A$ by 
		$x^{W_{\max}}$ and interpolate to get the polynomial corresponding to an entry
		-- then subtracting $\ell W_{\max}$ from each exponent will reveal the original 
		Laurent polynomial. Henceforth we will not
		allude to this artifice for dealing with Laurent polynomials.}
	sublinear $\CLP$
	for all $\ell\le 2^{c\sqrt{\log n}}$.
	Observe that every monomial appearing in
	$A^\ell[s,t]$ corresponds to an $s$-$t$ walk using at most $\ell$
	edges, and the degree of the monomial equals the total weight of the
	walk. Therefore, the least degree appearing in $A^\ell[s,t]$ equals the
	minimum weight of an $s$-$t$ walk using at most $\ell$ edges. This least degree monomial can be interpolated in logspace through \cite{HAB}.
	Setting $\ell=2^{c\sqrt{\log n}}$ proves the claim.
\end{proof}

Let $G=(V,E)$ be a digraph, and $U\subset V$ be a multiset. Then for a
path $P$ of $G$, we say that $U$ hits $P$, if
$V(P)\cap U\ne \phi$. We shall use the following sparsification lemma, which is very similar to \cite[Theorem 4.1]{CDKMR}.

\begin{lemma}\label{drep}
	There is an algorithm $\mcU$ such that for any directed graph
	$G=(V,E)$ on $n$ vertices, a value $\gamma=\Omega(\log n)$, and seed
	$\sigma\in\{0,1\}^{(\eps+6)\log n}$ for $\eps>0$, satisfies:
	
	\begin{enumerate}
		\item $\mcU(G,\gamma,\sigma,\eps)$ returns a multiset
		      $U_\sigma\subset V$ of size
		      $\mcO\Big(\frac {n\log n}\gamma\Big)$, in logspace
		      
		\item For any set of at most $n^2$ paths each of length $\gamma$, namely
		      $\mcP=(P_1,\ldots,P_{n^2})$ of $G$, we have
		      \[
			      \Pr_{\sigma\in\{0,1\}^{(\eps+6)\log n}}
			      [\forall i\le n^2: U_\sigma\text{ hits }P_i]>3/4.
		      \]
	\end{enumerate}
\end{lemma}

We push the proof of the above to \Cref{ap1}, it is essentially the same as that of \cite[Theorem 4.1]{CDKMR}. We now prove our main result.

\begin{algorithm}
	\caption{$\opn{DIST}_n$}\label{nl2:alg1}
	\begin{algorithmic}[1]
		\State \textbf{Input:} $G=(V,E),|V|=n,s,t\in V$.
		\State Fix $\eps>6$, $\lambda\gets2^{100\sqrt{\log n}}$.
		\State $\mu\gets\infty$
		\For {$\sigma\in\{0,1\}^{(\eps+6)\log n}$}
		\State $U_\sigma\gets\mcU(G,\lambda/2,\sigma,\eps)$
		\State $V_\sigma\gets \{s,t\}\cup U_\sigma$
		\State $E_\sigma$ is defined by the following loop:
		\For {$u,v\in V_\sigma$}
		\State $(u,v)\in E_\sigma$ if there is a path from $u$ to $v$
		using at most $\lambda$ edges
		\State $W_\sigma(u,v)\gets$ min weight of a path from
		$u$ to $v$ using $\le\lambda$ edges
		\Comment{\Cref{short}}
		\EndFor
		\State $G_\sigma\gets(V_\sigma,E_\sigma)$
		\State $A_\sigma\gets$ weighted adjacency matrix of $G_\sigma$
		with self-loops.
		\State $\mu_\sigma\gets$ degree of the minimum-degree term in
		$A_{\sigma}^{|V_\sigma|}[s,t]$
		\Comment{$\mu_\sigma=\infty$ if no such term exists.}
		\State $\mu\gets\min(\mu,\mu_\sigma)$
		\EndFor
		\State \Return $\mu$
	\end{algorithmic}
\end{algorithm}

\begin{proof}[Proof of \Cref{thm:dist}]
	We describe the algorithm in \Cref{nl2:alg1}, and discuss the proof of
	correctness here.
	
	Fix $\lambda=2^{c\sqrt{\log n}}$ for some constant $c$ to be chosen
	later. Let $\eps>6$ be fixed, and let
	$U_\sigma=\mcU(G,\gamma=\lambda/2,\sigma,\eps),$
	where $\sigma\in\{0,1\}^{(\eps+6)\log n}$ is as in
	\Cref{drep}.
	
	Define
	$
		V_\sigma=\{s,t\}\cup U_\sigma,
	$
	and let $G_\sigma$ be a graph on vertex set $V_\sigma$ where
	$
		(u,v)\in E_\sigma
	$
	iff there is a walk from $u$ to $v$ in $G$ containing at most
	$\lambda$ edges.
	
	Notice that
	$|V_\sigma|
		=
		\mcO\Big(\frac{n\log n}{\gamma}\Big).$
	
	Let
	$
		\rho_{st}=(s=u_0,u_1,\ldots,u_\ell=t)
	$
	be a minimum-weight simple path from $s$ to $t$.
	Such a path exists whenever
	$t$ is reachable from $s$, and satisfies $\ell\le n-1$.
	
	Divide $\rho_{st}$ into consecutive segments of $\gamma$ edges
	(except the last which is at most $\gamma$).
	Let
	$
		\ell'
		=
		\Big\lfloor\frac{\ell}{\gamma}\Big\rfloor.
	$
	For $i\in\{0,1,\ldots,\ell'-1\}$, let $\rho^i$ denote the segment between vertices
	$w_i = u_{\gamma\dot i}$ and $w_{i+1} = u_{\gamma(i+1)}$
	consisting of exactly $\gamma$ edges. Let $\rho^{\ell'}$ consist of the edges 
	after $w_{\ell'}$ i.e. $u_{\ell - {\ell}\bmod{\gamma}},\ldots,u_\ell = t$ -- 
	which consists of less than $\gamma$ edges.
	
	Let $\mcP=\{\rho^i\}_{0\leq i< \ell'}$.
	By \Cref{drep}, with probability at least $3/4$, every segment
	$\rho^i$ contains a vertex $v_i\in U_\sigma$ (for $i \in {0,..,\ell'-1}$).
	
	For consecutive vertices $v_i,v_{i+1}$, the subpath of $\rho_{st}$
	between them uses at most $2\gamma=\lambda$ edges. Hence
	$
		(v_i,v_{i+1})\in E_\sigma.
	$
	Notice that the number of edges between $s,v_0$ is at most $\gamma$
	and the number of edges between $v_{\ell'-1},t$ is at most $2\gamma$. 
	All other pairs $v_i,v_{i+1}$ are separated by at most $2\gamma$ edges.
	$
		(s,v_0,\ldots,v_{\ell'-1},t)
	$
	is a path in $G_\sigma$.
	Assign weights to edges of $G_\sigma$ by
	$
		W_\sigma:E_\sigma\to\mathbb \mathbb{Z},
	$
	where $W_\sigma(u,v)$ is the minimum weight of a path from $u$ to $v$
	using at most $\lambda$ edges.
	
	The subpath of $\rho_{st}$ between any two consecutive vertices among
	$
		s,v_0,\ldots,v_{\ell'-1},t
	$
	uses at most $\lambda$ edges. Hence the corresponding edge of
	$G_\sigma$ has weight at most the weight of that subpath.
	Summing over all segments shows that the minimum-weight path from
	$s$ to $t$ in $G_\sigma$ has weight at most the minimum-weight path
	weight in $G$.
	
	Conversely, every edge $(u,v)\in E_\sigma$ corresponds to a path in
	$G$ from $u$ to $v$ of weight exactly $W_\sigma(u,v)$. Therefore any path in
	$G_\sigma$ can be expanded into a walk in $G$ of the same total weight
	from which deleting cycles (necessarily non-negative) yields a path of weight
	at most the weight of the original compressed path.
	
	It follows that the minimum-weight path from $s$ to $t$ in $G_\sigma$
	has weight at least the minimum-weight path weight in $G$.
	Hence the minimum-weight path from $s$ to $t$ in $G_\sigma$ has exactly
	the same weight as the minimum-weight path from $s$ to $t$ in $G$.
	
	Let $A_\sigma$ be the weighted adjacency matrix of $G_\sigma$ given by
	\[
		A_\sigma[u,v]
		=
		\begin{cases}
			x^{W_\sigma(u,v)} & u \neq v \text{and} (u,v)\in E_\sigma, \\[2mm]
			1                 & u=v,                                   \\[2mm]
			0                 & \text{otherwise.}
		\end{cases}
	\]
	
	Using \Cref{short}, we can compute $A_\sigma$ in sublinear $\CLP$.
	
	We now compute
	$
		A_\sigma^{|V_\sigma|}[s,t].
	$ Let this polynomial be $p_\sigma(x)$.
	The least degree appearing in $p_\sigma(x)$ equals the minimum weight
	of an $s$-$t$ path in $G_\sigma$.
	Suppose this degree is $\mu_\sigma$. If
	$p_\sigma\equiv0$, we set $\mu_\sigma=\infty$.
	We enumerate all seeds
	$
		\sigma\in\{0,1\}^{(\eps+6)\log n}
	$
	and return
	$
		\mu=\min_\sigma \mu_\sigma.
	$
	For at least one good seed, the previous argument shows that
	$\mu_\sigma$ equals the minimum weight of an $s$-$t$ path in $G$.
	Therefore $\mu$ is exactly the minimum weight of an $s$-$t$ path in
	$G$.
	The catalytic space used is
	\[
		\mcO(|U_\sigma|\poly(\log n))
		=
		\dfrac n {2^{\Theta(\sqrt{\log n})}}
	\]
	for a sufficiently large choice of $c$.
	The workspace is clearly $\mcO(\log n)$, and the running time is
	polynomial.
\end{proof}

\begin{corollary}\label{cor:negCycle}(Negative cycle detection)
	Given an $n$-vertex directed graph $G=(V,E)$ with polynomially bounded integer edge weights, we can decide whether $G$ contains a negative-weight cycle in sublinear $\CLP$.
\end{corollary}

\begin{proof}
	We use the same construction as in Theorem~\ref{thm:dist}, building for each seed $\sigma$ the compressed graph
	$G_\sigma=(V_\sigma,E_\sigma)$ with weights $W_\sigma(u,v)$ defined as the minimum weight of a $u$-$v$ walk using at most $\lambda$ edges.
	But before that if $\mathrm{tr}(A^\lambda)$ already contained a monomial with
	a negative exponent then we have detected a negative cycle.
	
	Let $A_\sigma$ be the weighted adjacency matrix:
	\[
		A_\sigma[u,v] =
		\begin{cases}
			x^{W_\sigma(u,v)} & (u,v)\in E_\sigma, \\
			1                 & u=v,               \\
			0                 & \text{otherwise.}
		\end{cases}
	\]
	
	Consider
	$p_\sigma(x) = \operatorname{tr}(A_\sigma^{|V_\sigma|}),$
	i.e., the sum of diagonal entries of $A_\sigma^{|V_\sigma|}$.
	
	\begin{observation}
		Every monomial in $p_\sigma(x)$ corresponds to a closed walk in $G_\sigma$ of length at most $|V_\sigma|$, and the degree of the monomial equals the total weight of the walk.
	\end{observation}
	\begin{claim}
		$G$ contains a negative-weight cycle iff for some seed $\sigma$, the polynomial 
		$p_\sigma(x)$ contains a monomial of negative degree \emph{or}
		$\operatorname{tr}(A^\lambda)$ contains one such monomial where $A$ is the adjacency
		matrix of the original graph.
	\end{claim}
	\begin{proof}	
		\textbf{(1) If a negative cycle exists.}
		Let $C$ be a negative-weight cycle in $G$.
		By the same hitting-set argument as in Theorem~\ref{thm:dist}, with probability at least $3/4$, all segments of $C$ are represented in $G_\sigma$, so $C$ induces a corresponding closed walk in $G_\sigma$ with total weight upper 
		bounded by that of $C$.
		
		Since $C$ has negative weight, repeating it yields closed walks of negative total weight. In particular, within the length bound $|V_\sigma|$, there exists a closed walk whose corresponding monomial in $p_\sigma(x)$ has negative degree, since self loops are allowed.

		\textbf{(2) If no negative cycle exists.}
		Then every cycle in $G$ has nonnegative weight. Hence every closed walk decomposes into simple cycles of nonnegative total weight, implying every closed walk has nonnegative weight. Therefore all monomials in $p_\sigma(x)$ have nonnegative degree.
	\end{proof}
	
	Thus no negative-degree monomial appears.
	Therefore, $G$ contains a negative cycle iff for some seed $\sigma$, the polynomial
	$p_\sigma(x)$ contains a monomial of negative degree or there is a short
	negative cycle containing at most $\lambda$ edges.
	We enumerate all seeds $\sigma \in \{0,1\}^{(\epsilon+6)\log n}$ and accept iff such a negative-degree monomial exists.
	
	The complexity is identical to Theorem~\ref{thm:dist}.
\end{proof}
\begin{corollary}[Shortest Path]
	There is an algorithm which, given a positively weighted directed graph 
	$G=(V,E)$, polynomially bounded $W:E\to \bbN_{>0}$ and vertices $s,t \in V$, 
	outputs a shortest $s$ to $t$ path (if it exists) in sublinear $\CLP$.
\end{corollary}
\begin{proof}
	We shall in fact show how to find the lexicographically first $s$ to $t$ shortest path. Suppose the lex-first shortest path is $(v_0=s,v_1,\ldots,v_{\ell-1},v_{\ell}=t)$. At the $i^{th}$ step, we store $v_i$ in the worktape. Now, to find $v_{i+1}$, we enumerate over all out-neighbors of $v_i$ and pick the lex-first $u$ such that $W(v_i,u)+\opn{dist}(u,t)=\opn{dist}(v_i,t)$, and set $v_{i+1}\gets u$. To compute $\opn{dist}$, we use \Cref{thm:dist}, and thus we have the desired time-space bound to find the path.
\end{proof}

\section{Applications} \label{sec:app}
In this section we give applications related to determinant and matchings in graphs.
\subsection{Determinant in near linear $\CLP$}
We start by giving a space-space tradeoff algorithm for determinant, using Newton's identities and in fact, using a formula relating the ordinary generating function (OGF) of elementary
symmetric polynomials to the power polynomials\footnote{Actually
to $\sum_{n=1}^{\infty}{e^{\frac{(-1)^{n-1}}{n}p_nt^n}}$.}.
Our starting point is the following well-known identity which we establish in
\Cref{app:det} for the sake of completeness.

\begin{equation} \label{eq:det}
	\det(A)=
	[t^n]\exp\left(\sum_{k=1}^{n}(-1)^{k+1}\frac{p_k}{k}t^k\right)
	\pmod {t^{n+1}}.
\end{equation}

This brings us to the following.

\begin{theorem}\label{thm:det}
	Let \(A\) be an \(n\times n\) matrix over integers. Given \(A,A^2,\ldots,A^n\), the value \(\det(A)\) is computable in DLOGTIME-uniform \(\TC^0 \subseteq \Log\).
\end{theorem}

\begin{proof}
	From the given powers we first compute
	\[
		p_k=\operatorname{tr}(A^k), \qquad 1\le k\le n,
	\]
	which is just an iterated sum and hence is in uniform \(\TC^0\). By
	\cref{eq:det},
	\[
		\det(A)
		=
		[t^n]\sum_{j=0}^n{\frac{\left(\sum_{k=1}^{n}(-1)^{k+1}\frac{\operatorname{tr}(A^k)}{k}t^k \right)^{j}}{j!}}
	\]
	Thus the determinant is the coefficient of $t^n$ in a polynomial of
	degree $n^2$ over rationals.
	So we need to plug in
	$n^2 + 1$ distinct values for $t$ and do Lagrange interpolation to
	recover the coefficient of $t^n$.
	Consider the rationals that occur in the expression for the
	polynomial and their denominators (namely $j!,k$ for $j,k \leq n$) and observe that none of
	the denominators has a prime factor exceeding $n$.
	Moreover the numerator and denominator of any rational whose sum
	constitutes a coefficient of $t^n$ is bounded by $2^{n^c}$ for some
	constant $c$.
	Thus we can do Chinese remaindering with polynomially many primes each
	of which
	is polynomially bounded in $n$ \emph{and} greater than $n^2$ precluding
	it from dividing any denominator. Since the determinant is an integer,
	the process of Chinese Remaindering modulo the primes will ensure that
	we obtain the correct determinant. We will use \cite{HAB} for this,
	placing it in DLOGTIME-uniform $\TC^0$ which is contained in $\Log$,
	completing the proof
	\footnote{For our purposes the weaker Logspace-uniform $\NC^1$ bound
		proved by \cite{CDL} would also have sufficed.}.
\end{proof}
Therefore we have the following as corollary.
\begin{corollary}\label{cor:det}
	Let \(A\) be an \(n\times n\) matrix over integers, and $k\in[1, n]$ a parameter. Then, $\det(A)$ can be computed in
	\[
		\C\TISP\left[\poly(n^{\log (3k)}),\mcO(\log n \cdot \log (3k)), \mcO\left(\frac{n}{k} \cdot \log ^2 n\right)\right].
	\]
\end{corollary}
\begin{proof}
	This follows from \Cref{thm:apow,thm:det}, by combining space-space tradeoff of matrix powering and determinant coefficient-extraction.
\end{proof}

\begin{remark}
	We note that the determinant can also be shown in near linear $\CLP$ using Csanky's algorithm, see \cite{Kozen1992}, or by Berkowitz's method \cite{Berkowitz84}.
\end{remark}

We now give an analogous space-space tradeoff algorithm for computing Pfaffian by relating it to the determinant.
\subsection{Pfaffian of a Matrix in near linear $\CLP$}
We need the following identity which we prove in \Cref{app:pfaff},
\begin{equation} \label{eq:pfaff}
	\pf(A) = [t^m] \exp\left(\frac{1}{2} \cdot \sum_{k = 1}^{m} (-1)^{k+1}\frac{\operatorname{tr}(B^k)}{k}t^k\right) \pmod{t^{m+1}},
\end{equation}
where $m = n/2$, for some even integer $n$, and $B = -JA$. Here $J$ is a fixed skew symmetric integer matrix depending only on $n$, see \Cref{app:pfaff} for more details. Using this notation we have the following.
\begin{theorem}\label{thm:pfaff}
	Let $A$ be a $n \times n$ skew symmetric matrix over integers. Given \(A, B, B^2, \dots, B^m\), the value $\pf(A)$ is computable by DLOGTIME-uniform \(\TC^0\) circuits.
\end{theorem}
\begin{proof}
	The proof is similar to the proof of \Cref{thm:det}.
\end{proof}

Thus, we get the following as corollary.
\begin{corollary}\label{cor:pfaff}
	Let \(A\) be a skew symmetric \(n\times n\) matrix over integers, with $n$ even, and $k\in[1, n]$ a parameter.
	Then, $\pf(A)$ can be computed in
	\[
		\C\TISP\left[\poly(n^{\log (3k)}),\mcO(\log n \cdot \log (3k)), \mcO\left(\frac{n}{k} \cdot \log ^2 n\right)\right].
	\]
\end{corollary}
\begin{proof}
	This follows from \Cref{thm:apow,thm:pfaff}, by combining space-space tradeoff of matrix powering and Pfaffian coefficient-extraction.
\end{proof}

We give one application on matchings which uses the computation of Pfaffian.
\subsection{Planar Perfect Matching in near linear $\CLP$}
\begin{restatable}{theorem}{PlMinWTPM}\label{thm:PlMinWtPM}
	Given a planar graph $G=(V,E)$ on $n$ vertices, and a polynomially bounded weight assignment $W_{input}:E\to [n^c]$, we can compute a minimum weight perfect matching of $G$ in near linear $\CLP$.
\end{restatable}
\begin{proof}
	Notice first that \cite{Kastelyn,MSV} and \Cref{cor:pfaff} (setting $k=1$) immediately give a way to count the number of min-weight perfect matchings of $G$ in near linear $\CLP$.

	We conclude the proof by the reduction from search to weighted decision
	as described in \cite{ACD}. For a self-contained proof
	see \Cref{app:MNProof}.
\end{proof}

Our second application concerns planar bipartite perfect matching which builds on \Cref{sec:subd}.
\subsection{Planar Bipartite Perfect Matching in sublinear $\CLP$} \label{subsec:mn}
Let $G$ be a planar bipartite graph. Miller \& Naor \cite{MN} (see also
\cite{DGKT}) had proved
that detecting a perfect matching in an $n$ node graph reduces to
detecting the absence of a negative cycle in a weighted version of the
dual graph. Further, if there is no such negative cycle then shortest
distances in the weighted dual yield an explicit representation of a
perfect matching. Since the dual graph also has $\mcO(n)$ vertices,
from Euler's relation, this yields a sublinear catalytic space algorithm for
decision and search for bipartite planar perfect matching. The following is the
version of the theorem proved in \Cref{app:MN}.

\begin{theorem}[Miller \& Naor \cite{MN}]\label{thm:MN}
    Given a planar flow network
    $(G = (V,E),c: E \rightarrow \mathbb{N}\cup\{0\}, d: V \rightarrow \mathbb{Z})$
    where the image of $c,d$ are in unary, the detection of a feasible flow
    and its computation can be Logspace reduced to finding the existence of a
    negative cycle in a planar digraph on $O(|V|)$ vertices
    with unary integer weights and to distance
    in a planar digraph on $O(|V|)$ vertices without negative cycles with unary
    integer weights.
\end{theorem}

Therefore we get the following.
\begin{theorem}\label{thm:MWPlBPM}
    There is an algorithm that, given a bipartite planar graph
    on $n$ vertices, can find a perfect matching
    of the graph in sublinear $\CLP$.
\end{theorem}
\begin{proof}
    Immediate from \Cref{thm:dist,cor:negCycle}, sublinear weighted distance and negative cycle detection, combined with \Cref{thm:MN}.
\end{proof}

\section{Conclusion}
We have shown near linear catalytic log space algorithms for determinant computation, planar matchings, and sublinear catalytic logspace algorithms for distances, and bipartite planar graphs. Couple of interesting questions in these lines are:\begin{itemize}
	\item Can one design sublinear space algorithms that simultaneously run in polynomial time, catalytic or otherwise, for determinant computation.
	\item Can we have near linear catalytic logspace algorithms for general matchings; note that general matchings have recently been put in $\CL$, \cite{CDKMS}.
\end{itemize}

\section*{AI disclosure}
The Pfaffian identity \cref{eq:pfaff} was discovered through ChatGPT while
trying to extend \Cref{thm:det} to compute the Pfaffian. We realised
that Berkowitz's method \cite{Berkowitz84} can also be adapted to give a
near linear $\CLP$ through Chat GPT. Chat GPT Astra was also used for routine
formatting and proof verification. The remaining manuscript remains human
generated and we take full responsibility for its correctness.

\bibliographystyle{alpha}
\bibliography{biblio}

\appendix
\crefalias{section}{appendix}

\section{\texorpdfstring{Proof of \Cref{drep}}{Proof of Lemma~\ref{drep}}}\label{ap1}
An undirected graph $D$ is an $(N,\Delta,\alpha)$-expander graph if it is a $\Delta$-regular connected graph on $N$ vertices, and the second largest eigenvalue (the absolute value) of its normalized adjacency matrix is bounded by $\alpha$ ($0\le \alpha\le 1$).

\begin{lemma}[\cite{lubotzky1988ramanujan,morgenstern1994existence}]\label{expc}
    For every prime $p$ and every positive integer $k$, there exists an infinite family of expander graphs $\left\{D_N\right\}_N$, such that each $D_N$ is an $(N, \Delta, \alpha)$-expander for $\Delta=p^k+1$ and $\alpha \leq \frac{2 \sqrt{\Delta-1}}{\Delta}$. Furthermore, there is a logspace algorithm (with simultaneous polylogarithmic runtime) that, given, $v\in\{0,1\}^{\log N}$ and $i\in[\Delta]$, computes the $i^{th}$ neighbor of $v$ in $D_N$.
\end{lemma}
\begin{definition}[Expander random walk]
    For any $(N, \Delta, \alpha)$-expander $D=\left(V_D, E_D\right)$, an expander random walk of length $K-1$ is $\left(X_1, \ldots, X_K\right)$ where $X_1$ is the uniform distribution over $[N]$, and for every $i>1, X_i$ is the uniform distribution over the $\Delta$-sized neighborhood of the $(i-1)^{\text {th }}$-vertex in the random walk.
\end{definition}

Expander random walks of length $K$ over $[N]$ behaves similar to uniformly sampling $K$ elements from $[N]$, as follows:
\begin{lemma}\label{exw}
    Let $D$ be an $(N, \Delta, \alpha)$-expander, and let $\mathcal{B} \subset V_D$ with $\beta=\dfrac{|\mathcal{B}|}{N}$. Then, for every $K$, the probability that an expander random walk $\left(X_1, \ldots, X_K\right)$ of length $(K-1)$ in $D$ never leaves $\mathcal{B}$, is at most $\beta(\beta+\alpha \cdot(1-\beta))^{K-1}$.
\end{lemma}
\begin{lemma}[\cite{vadhan2012pseudorandomness}]
    For every $m, n \in \mathbb{N}$, there exists a pairwise independent hash function family $\mathcal{H}_{m, n}$, where a random function can be sampled from $\mathcal{H}_{m, n}$ using a seed $r$ of bitlength $\max \{\log m, \log n\}+ \log n$.
    
    Moreover, given any $h \in \mathcal{H}_{m, n}$ (described by a string in $\{0,1\}^r$ ) and $x \in[m], h(x)$ can be computed in time $\poly(\log m,\log n)$ and space $\mathcal{O}(r)$.
\end{lemma}
Next, we have the following lemma about a pairwise independent hash function family $\mathcal{H}_{m, n}$. We say that a hash function $h$ hits a set $B \subseteq[n]$, if $\operatorname{Im}(h)$ intersects $B$. \cite{CDKMR} shows that a randomly chosen $h$ from $\mathcal{H}_{m, n}$, where $m= \mathcal{O}(n / \gamma)$, hits any $\gamma$-sized subset of $[n]$ with large probability.
\begin{lemma}[\cite{CDKMR}]\label{hashi}
    For any constant $0<\beta<1 / 2$, let $m=\frac{n(1 / \beta)}{\gamma}$, and let $\mathcal{H}_{m, n}$ be a pairwise independent hash function family from $[m]$ to $[n]$. Then, for any $B \subset[n]$ of size $\gamma$ we have,
    
    $$
        \operatorname{Pr}_{h \sim \mathcal{H}_{m, n}}[h \text { hits } B] \geq 1-\beta .
    $$
    
\end{lemma}
Finally, we prove \Cref{drep}.

\begin{proof}[Proof of \Cref{drep}]
    For a fixed $\gamma=\Omega(\log n)$ and small enough $\beta\in(0,1/2)$ to be fixed later, let $m=\frac{n(1 / \beta)}{\gamma}$, $N=|\mcH_{m,n}|$, and $D=(V_D,E_D)$ be an $(N,\Delta=1+1/\beta^2,\alpha)$-expander given by \Cref{expc} (assume that $\beta^{-1} $ is a prime power). Observe that by the choice of $\Delta$ and the fact that $\alpha \leq \frac{2 \sqrt{\Delta-1}}{\Delta}$, we have $\alpha\le 2\beta$.
    
    In $\mcU$, $U_\sigma$ is constructed as follows: Use $\log N$ bits to identify a vertex $h_1$ of $D$, this is like uniformly sampling from $X_1$ which is the random variable for a vertex of $D$. Next, sample a $K-1$-expander random walk using $K\log \Delta$ bits to get $(h_1,\ldots,h_K)\sim(X_1,\ldots,X_K)$. This takes $\log N+K\log\Delta$ bits. $\log N\le 2\log n$ since $m<n$ for small enough $\beta$. Again it can be seen (as in \cite{CDKMR}) that for small enough $\beta$, setting $K=\frac {(2+\frac\eps8)\log n}{\log(\frac 1{3\beta})}$, we have $2\log n+K\log \Delta\le (\eps+6)\log n$ i.e. the seed $\sigma \in\{0,1\}^{(\eps+6)\log n}$. Finally, $U_\sigma=\cup_{i\in[K]}\opn{Im}h_i$, and thus we have $|U_\sigma|\le K\cdot m=\mcO\Big(\frac {n\log n}\gamma\Big)$.

    Let $\mathcal{B}$ be the subset of $V_D$ in $D$ that contains hash functions in $\mathcal{H}_{m, n}$ that do not hit $P_{1}$. Applying  \Cref{hashi} for subsets of size $\gamma$, we see that $|\mathcal{B}|/N\le \beta$.
    
    Now, using the property of an expander random walks, \Cref{exw} and the fact that $\alpha\le 2\beta$, we have
    
    $$
        \begin{aligned}
            \underset{\left(h_1, \ldots, h_K\right) \sim\left(X_1, \ldots, X_K\right)}{\operatorname{Pr}} & {\left[\forall i \in[K], \opn{Im}(h_i)\cap V(P_1)=\phi\right] }       \\
                                                                                                          & \leq \beta \cdot(\beta+\alpha \cdot(1-\beta))^{K-1}                   \\
                                                                                                          & \leq \beta \cdot(\beta+2 \beta \cdot(1-\beta))^{K-1} \leq(3 \beta)^K
        \end{aligned}
    $$

    Again, since $K=\frac{\left(2+\frac{\varepsilon}{8}\right) \log (n)}{\log \left(\frac{1}{3 \beta}\right)}$, the probability that $P_{1}$ does not intersect with $U_\sigma$ is at most $\frac{1}{n^{2+\varepsilon / 8}}$. Applying the union bound over all paths $P_i\in\mcP$, we have that:
    \[\Pr_{\sigma\in\{0,1\}^{(\eps+6)\log n}}[\forall i\le n^2: U_\sigma\cap V(P_i)\ne \phi]\ge 1-\frac{1}{n^{\varepsilon / 8}}>3/4.\]
    
\end{proof}

\section{The determinant and Pfaffian as the coefficient of a polynomial}
\subsection{Determinant}\label{app:det}
Newton's identities give us a way to compute elementary symmetric polynomials from power sums and vice-versa.
Let $x_1, \dots, x_n$ be variables, denote for $k \geq 1$ by $p_k(x_1, \dots, x_n)$ the $k$-th power sum:
\[
	p_k(x_1, \dots, x_n) = \sum_{i=1}^n x_i^k = x_1^k + \dots + x_n^k,
\]
and for $k \geq 0$ denote by $e_k(x_1, \dots, x_n)$ the elementary symmetric polynomial, that is, the sum of all distinct products of $k$ distinct variables, so
\begin{align*}
	e_0(x_1, \dots, x_n) & = 1,                                  \\
	e_1(x_1, \dots, x_n) & = x_1 + x_2 + \dots + x_n,            \\
	e_2(x_1, \dots, x_n) & = \sum_{1 \leq i < j \leq n} x_i x_j, \\
	                     & \vdots                                \\
	e_n(x_1, \dots, x_n) & = x_1 x_2 \cdots x_n,                 \\
	e_k(x_1, \dots, x_n) & = 0, \quad \text{for } k > n.
\end{align*}

Let $\lambda_1, \dots, \lambda_n$ be the eigenvalues of the matrix $A$. We know that sum of eigenvalues of a matrix is equal to its trace and eigenvalues of $A^\ell$ are the eigenvalues of $A$ raised to the power $\ell$, for more details refer to any standard linear algebra textbook.

Then from the above two facts, we get $\operatorname{tr}(A^\ell) \equiv \sum_{i=1}^{n}\lambda_i^\ell \pmod{q}$, which we can compute using the matrix powering algorithm from \ref{thm:apow}. Thus $\sum_{i=1}^{n}\lambda_i^\ell$, for $1 \leq \ell \leq n$, are our power sums and from these we will compute the $n^{th}$ elementary symmetric polynomial $\prod_{i=1}^{n} \lambda_i$, which is equal to the determinant of the matrix $A$.

The ordinary generating function of the elementary symmetric polynomials is
\[
	E(t):=\sum_{j=0}^n e_j(\lambda_1,\ldots,\lambda_n)t^j
	=\prod_{i=1}^n(1+\lambda_i t)
	=\det(I+tA).
\]

Taking logarithms as formal power series gives
(see Stanley~\cite[Chapter 7, Equation~(7.23)]{Stanley99ec2})
$$\log E(t)
	= \sum_{i=1}^n \log(1+\lambda_i t)
	= \sum_{i=1}^n \sum_{k\geq 1}(-1)^{k+1}\frac{\lambda_i^k t^k}{k}
	= \sum_{k\geq 1}(-1)^{k+1}\frac{p_k}{k}t^k .$$
Hence we have,
$$
	\det(I+tA)=
	\exp\left(\sum_{k\geq 1}(-1)^{k+1}\frac{\operatorname{tr}(A^k)}{k}t^k\right).
$$
Since $\det(I+tA)$ has degree at most $n$, the coefficient of $t^n$ is obtained from the truncation modulo $t^{n+1}$, thus yielding \cref{eq:det}.

\subsection{Pfaffian}\label{app:pfaff}
Fix the $2\times 2$ matrix:
$$
	J_2=\begin{pmatrix}0&1\\-1&0\end{pmatrix}
$$
and $J$ be a block diagonal matrix of size $n \times n$ as
$
	J=\operatorname{diag}(J_2,\ldots,J_2).
$
We have
$
	\pf(J)=1, \det(J)=1$ and $J^{-1}=-J.
$

Consider
$
	Q(t)=\pf(J+tA),
$
which is valid since the matrix $J+tA$ is skew symmetric too. Every term in a Pfaffian is a product of \(m\) matrix entries. Therefore \(Q(t)\) has degree at most \(m\), and choosing the \(tA\) contribution from every factor gives $[t^m]Q(t)=\pf(A)$.

By Cayley's theorem
$
	Q(t)^2=\det(J+tA).
$
Set $B=J^{-1}A=-JA$. Then we have
$$
	Q(t)^2=\det J\cdot\det(I+tB)=\det(I+tB).
$$

Since $Q(0) = \pf(J) = 1$, it is the unique formal square root with constant term $1$, thus we have
$$
	Q(t)=\sqrt{\det(I+tB)}.
$$

From the determinant identity we derived in the previous section, we get
$$
	\sqrt{\det(I+tB)}=
	\exp\left(\frac{1}{2} \cdot \sum_{k\geq 1}(-1)^{k+1}\frac{\operatorname{tr}(B^k)}{k}t^k\right)
$$
and consequently we get \cref{eq:pfaff}.

\section{\texorpdfstring{Remaining part of Proof of \cref{thm:PlMinWtPM}}{Remaining part of Proof of Theorem~\ref{thm:PlMinWtPM}}}\label{app:MNProof}
\PlMinWTPM*
\begin{proof}
	Now, consider an extra block of catalytic space consisting of $100m\log m$ bits. A compress or compute algorithm to compute a min-weight perfect matching is as follows (this is a careful adaptation of \cite{ACD}'s reduction from weighted decision to weighted search):
	\begin{enumerate}
		\item  Let $free\gets 0$ be the amount of free space available to us from the catalytic tape.
		\item Read the first $10m\log m $ bits after $free$ number of bits on the catalytic tape, as a weight assignment: $W(e_1),\ldots,W(e_m)$, $W:E\to[m^{10}]$. Check if the weight assignment $\hat W=MW_{input}+W$ min-isolates a perfect matching in $G$ (use the counting oracle to do this, done in \cite{ACD}), where $M = nm^{10}$ is a large constant.
		\item Suppose $\hat W$ is indeed a min-isolating weight assignment. Then output the unique min-weight matching w.r.t $\hat W$, as the set of edges on deleting which the minimum value changes (done in detail in \cite{ACD}) and go to step $5$ for reconstructing the catalytic tape.
		\item Suppose $W$ is not isolating, then we find a threshold edge $e_i$ (use \cite{ACD}), and replace $W$ with $(freespace,i,W(e_1),\ldots,W(e_{i-1}),W(e_{i+1}),\ldots,W(e_m))$ where $free\gets free+9\log m$ is currently the number of free bits available (which is $10m\log m-\log m-10(m-1)\log m$), and $freespace$ is just $free$-many bits available to us for any computation. At this stage, either $free\le 100m\log m-10m\log m$ or not. If $free\le 90m\log m$, repeat from step $2$. Otherwise do the following:
		      Using the oracle for weighted counting as a decision oracle, search for a min-weight perfect matching using self-reducibility of matching. Clearly, self-reduction requires to store $m$ edges, i.e. uses at most $m\log m$ space. This can be done in the $freespace$. Finally go to step $5$ to reconstruct the catalytic tape.
		\item To reconstruct the catalytic tape, we decompress the weight assignments backwards, i.e, consider the first $\log m+10(m-1)\log m$-bits after the $freespace$ ($free$ number of bits): $(i,W(e_1),\ldots,W(e_{i-1}),W(e_{i+1}),\ldots,W(e_m))$. Compute the weight of the threshold edge $e_i$ (\cite{ACD}), and replace it with $(W(e_1),\ldots,W(e_i),\ldots,W(e_m))$ and subtract $free\gets free-9\log m$. Repeat this process, until $free=0$, i.e. when the catalytic tape has finally been restored.
	\end{enumerate}
	To fill in the details, we refer \cite{ACD}. The above clearly requires $\mcO(n\log^2 n+m\log n)=\mcO(n\log^2 n)$ catalytic space, and runs in polynomial time. Thus, we have the result.
\end{proof}

\section{An overview of Miller \& Naor in our context}\label{app:MN}
We now make the sketch from \Cref{subsec:mn} more precise. Miller \& Naor \cite{MN} deal
more generally with the problem of finding a feasible flow in a capacitated
planar network with designated demands for every vertex. In the case of
planar bipartite matching, every vertex in one bipartition $A$ has a negative
demand of $-1$ (i.e. a surplus of $1$) and every vertex in the other bipartition
$B$ a positive demand of $+1$. On the other hand every arc $(a,b)$ has capacity
$c(a,b) = 1$ while the reverse arc $(b,a)$ has a capacity $c(a,b) = 0$
$a \in A, b \in B$.

In order to construct such a flow function $f$, they first ignore the capacities
and find a flow $f^{(0)}$ in the network $G^{(0)}$ satisfying the demands.
To do so, they
construct a spanning tree $T$ in the network and make sure that every edge
$e = \{u,v\}$ carries a flow equal to the sum of the demands in the subtree
containing $v$ formed by deleting the edge $e$ from the tree. Notice that the
tacit assumption is that the total demand in the network is zero and hence the
flow on the edge matches the surplus on the $u$ side subtree of $T - e$. It
also means $f^{(0)}$ is skew-symmetric.

They then construct a new flow network $G^{(1)}$ with the capacity of edge $e$
being $c_1(e) = c(e) - f^{(0)}(e)$ (and for the reverse edge $e^R$,
$c_1(e^R) = c(e^R) - f^{(0)}(e^R) = c(e^R) + f^{(0)}(e)$). They go on to
construct the directed dual
of this network where the dual $e^*$ of the edge $e$ is obtained by turning
the edge $e$, clockwise. The weight $w(e^*) = c_1(e)$. They show that the
original flow network has a valid flow iff the dual flow network under the
weights $w$ has no negative cycle.

If the dual does not have a negative cycle, shortest paths are well defined
in the dual. Let $s^*$ be an arbitrary vertex in the dual (i.e. a face in the
primal). Compute the shortest distance $d(s^*,v^*)$ from $s^*$ to every
vertex $v^*$. Define $f^{(1)}(e) = d(s^*,v^*) - d(s^*,u^*)$ where
$(u^*,v^*)$ is the directed dual edge of $e$. Then, \cite{MN} show that
$f^{(1)}$ is a valid flow in the residual graph $c - f^{(0)}$ i.e.
$f^{(1)}$ is skew-symmetric and $f^{(1)}(e) \leq c(e) - f^{(0)}(e)$ for every
edge $e$. In other words, $f = f^{(0)} + f^{(1)}$ is a valid flow in the
original network $G$. This is because it is skew-symmetric from the
skew-symmetry of $f^{(0)}, f^{(1)}$ and satisfies
$f(e) = f^{(0)}(e) + f^{(1)}(e) \leq c(e)$.

Therefore in conclusion we obtain \Cref{thm:MN}.

\end{document}